\documentclass{article}

\usepackage[utf8]{inputenc} %
\usepackage[T1]{fontenc}    %
\usepackage[colorlinks=true,linkcolor=black,anchorcolor=black,citecolor=black,filecolor=black,menucolor=black,runcolor=black,urlcolor=black]{hyperref}       %
\usepackage{url}            %
\usepackage{booktabs}       %
\usepackage{amsfonts}       %
\usepackage{nicefrac}       %
\usepackage{microtype}     %
\usepackage{xcolor}         %
\usepackage{bm}         %
\usepackage{graphicx}  		%
\usepackage{amsmath} %
\usepackage{cite} %
\usepackage{diagbox}
\usepackage{siunitx}
\usepackage{multirow}
\usepackage{tabularx}
\usepackage{hhline}
\usepackage{arydshln}		%
\usepackage{xcolor}
\usepackage{mathtools} %
\usepackage{amsthm}			%
\usepackage{amssymb}			%
\usepackage{algorithmic}	%
\usepackage{algorithm}
\usepackage{pifont}
\usepackage{threeparttable}
\usepackage{fullpage}

\usepackage{bm}         %
\usepackage{amsthm}
\usepackage{xparse}
\usepackage{amssymb}
\usepackage{orcidlink}
\usepackage{enumitem}

\newtheorem{proposition}{Proposition}
\newtheorem{definition}{Definition}

\newtheorem{lemma}{Lemma}

\makeatletter %
\newtheorem*{rep@proposition}{\rep@title}
    \newcommand{\newrepproposition}[2]{%
    \newenvironment{rep#1}[1]{%
     \def\rep@title{#2 \ref{##1}}%
     \begin{rep@proposition}}%
     {\end{rep@proposition}}}
\makeatother %

\makeatletter %
\newtheorem*{rep@corollary}{\rep@title}
    \newcommand{\newrepcorollary}[2]{%
    \newenvironment{rep#1}[1]{%
     \def\rep@title{#2 \ref{##1}}%
     \begin{rep@corollary}}%
     {\end{rep@corollary}}}
\makeatother %

\newrepproposition{proposition}{Proposition}
\newrepcorollary{corollary}{Corollary}

\def\defn{\,\coloneqq\,}

\def\prox{{\mathsf{prox}}}

\def\max{{\mathsf{max}}}
\def\min{\mathop{\mathsf{min}}}

\def\defn{\,\coloneqq\,}

\def\R{\mathbb{R}}

\def\N{\mathbb{N}}

\def\dbm{{\bm{d}}}
\def\ebm{{\bm{e}}}
\def\gbm{{\bm{g}}}

\def\mbm{{\bm{m}}}

\def\sbm{{\bm{s}}}
\def\ubm{{\bm{u}}}

\def\wbm{{\bm{w}}}
\def\xbm{{\bm{x}}}
\def\ybm{{\bm{y}}}
\def\zbm{{\bm{z}}}

\def\deltabm{{\bm{\delta}}}
\def\betabm{{\bm{\beta}}}

\def\Abm{{\bm{A}}}

\def\Ibm{{\bm{I}}}
\def\Mbm{{\bm{M}}}

\def\Abm{{\bm{A}}}
\def\Dbm{{\bm{D}}}

\def\Ibm{{\bm{I}}}

\def\Wbm{{\bm{W}}}

\def\Ucal{{\mathcal{U}}}
\def\Tcal{{\mathcal{T}}}

\def\Scal{{\mathcal{S}}}

\def\Tsf{{\mathsf{T}}}

\def\Tsf{{\mathsf{T}}}

\def\operatorBoundConst{\epsilon_1}

\def\operatorSubdiffConst{\epsilon_2}

\def\sizeConst{n}
\def\dimConst{d}

\def\shrinkStepConst{2 \sqrt{ \dimConst }}
\NewDocumentCommand{\approxTV}{O{}}{\Scal^{#1}_{\tau}(\zbm)}

\NewDocumentCommand{\shrinkageFunction}{O{} O{} m}{\Tcal^{#1}_{#2}(#3)}

\def\TV{{h}}

\NewDocumentCommand{\helperh}{O{}}{\bar{h}^{#1}}

\def\approxh{\widehat{h}}

\def\argmin{\mathop{\mathsf{arg\,min}}} %

\definecolor{DarkGreen}{rgb}{0.0, 0.5, 0.0} %

\title{Closed-Form Approximation of the\\Total Variation Proximal Operator}

\author{
Edward P.\ Chandler, Shirin Shoushtari,\\
Brendt\ Wohlberg, and Ulugbek S.\ Kamilov\\
\small Washington University in St. Louis, MO 63130, USA\\
\small Los Alamos National Laboratory, Los Alamos, NM 87545, USA\\
\small \texttt{\{e.p.chandler, s.shirin, kamilov\}@wustl.edu, \texttt{brendt@lanl.gov} }
}

\begin{document}

\maketitle
\let\thefootnote\relax\footnote{This research was supported by the Laboratory Directed Research and Development program of Los Alamos National Laboratory under project numbers 20200061DR and 20230771DI.}

\begin{abstract}
Total variation (TV) is a widely used function for regularizing imaging inverse problems that is particularly appropriate for images whose underlying structure is piecewise constant.
TV regularized optimization problems are typically solved using proximal methods,
but the way in which they are applied is constrained by the absence of a closed-form expression for the proximal operator of the TV function.
A closed-form approximation of the TV proximal operator has previously been proposed, but its accuracy was not theoretically explored in detail. 
We address this gap by making several new theoretical contributions, proving that the approximation leads to a proximal operator of some convex function, it is equivalent to a gradient descent step on a smoothed version of TV, and that its error can be fully characterized and controlled with its scaling parameter.
We experimentally validate our theoretical results on image denoising and sparse-view computed tomography (CT) image reconstruction.
\end{abstract}

\section{Introduction}\label{section: introduction}
Inverse problems are typically posed as the estimation of a signal $\xbm \in \R^{\sizeConst}$ from measurements 
\begin{equation} \label{eq:fwd}
\ybm = \Abm \xbm + \ebm \;,
\end{equation}
where $\Abm \in \R^{m \times \sizeConst}$ denotes the \emph{forward operator} and $\ebm \in \R^m$ represents additive noise. 
The classical approach to solving such problems is to pose them as an optimization problem
\begin{equation} \label{eq: main optimization problem}
    \xbm^{\ast} \in \argmin_{\xbm \in \R^{\sizeConst}} f(\xbm) \quad \text{with} \quad f(\xbm) = g(\xbm) + \lambda h(\xbm) \;, 
\end{equation}
 where $g$  and $h$  represent the data-fidelity and regularization terms, respectively, and $\lambda \geq 0 $ is the regularization parameter.
 When the noise in~(\ref{eq:fwd}) corresponds to the additive white Gaussian noise (AWGN), it is common to use the squared $\ell_2$ data fidelity term $g(\xbm) = \frac{1}{2} \| \ybm - \Abm \xbm \|_2^2$.
One of the most widely used choices for $h$ is total variation (TV)~\cite{Rudin.etal1992}, which promotes solutions with a sparse image gradient. While it is no longer competitive for natural imagery, it remains effective for other types
of imagery, such as those that occur in scientific or industrial imaging, where the implicit piecewise-constant image model
is appropriate.

In this work, we focus on the two most common forms of TV, anisotropic $h^{\text{a}}$ and isotropic $h^{\text{i}}$, which are defined as
\begin{align}
    h^{\text{a}}(\xbm) & :=  \sum_{i=1}^{\sizeConst} \sum_{j=1}^{\dimConst} \bigr| [\Dbm_j \xbm]_i \bigr| \label{eq: aniso TV def} \\
    h^{\text{i}}(\xbm) & :=   \sum_{i=1}^{\sizeConst} \sqrt{\sum_{j=1}^{\dimConst} \bigr( [\Dbm_j \xbm]_i \bigr)^2},  \label{eq: iso TV def}
\end{align}
where $\Dbm: \R^{\sizeConst} \rightarrow \R^{\sizeConst \dimConst}$ is the discrete gradient operator, $\sizeConst$ denotes the signal size (e.g. $\sizeConst=1024$ in a $32 \times 32$ image), and  $\dimConst$ represents the signal dimension (e.g. $\dimConst=2$ for images). 
 
Gradient descent and related methods are not appropriate for problem~\eqref{eq: main optimization problem} when $h$ is the TV function since it is non-smooth, but \emph{proximal methods}~\cite{Parikh2014} are very effective. They are based on the \emph{proximal operator} of a function $h$, defined as
\begin{equation} \label{eq: proximal operator}
    \prox_{\tau h}(\zbm) := \argmin_{\xbm \in \R^{\sizeConst}} \Bigg\{  \frac{1}{2}\| \xbm - \zbm \|_2^2 + \tau h(\xbm)  \Bigg\},
\end{equation}
where $\tau >0$ is the proximal scaling parameter. This family of optimization algorithms includes the
accelerated proximal gradient method (APGM)~\cite{Beck.Teboulle2009} and the alternating direction method of multipliers (ADMM)~\cite{Boyd.etal2011}. 

Unfortunately, there is no closed-form solution to the proximal operator~\eqref{eq: proximal operator} of the TV function, requiring it to be computed using an iterative approach~\cite{Beck.Teboulle2009a} within algorithms such as APGM. These sub-iterations can be avoided via an appropriate ADMM variable splitting strategy~\cite{goldstein-2009-split}, but in many cases this just pushes the computational complexity to a different component of the algorithm.

A closed-form approximation to the TV proximal operators was introduced in~\cite{Kamilov.etal2014a, Kamilov2016, Kamilov2017, Kamilov2016a}, based on the concept of proximal average~\cite{Yu2013}.
While the effectiveness of this approximation was analyzed within the context of APGM, its more general properties as an approximation to the TV proximal operator were not explored in detail, and its application in other proximal algorithms was not considered.
In this paper, we address this gap by presenting new theoretical results demonstrating that the operator is indeed a good approximation to the proximal operators of anisotropic and isotropic TV. 
More specifically, we present two new contributions:
\begin{enumerate}[label=(\arabic*)]
    \item We provide new theoretical justifications for the approximation in~\cite{Kamilov.etal2014a, Kamilov2016, Kamilov2017, Kamilov2016a}.
    Specifically, we show that it is the proximal operator of some convex function, that it always decreases the TV function, and that its accuracy depends on the scaling parameter of the approximate proximal operator.
    \item We provide new numerical results that validate our theoretical analysis by using the approximate TV proximal operator within APGM and ADMM algorithms for limited angle computed tomography reconstruction. These numerical results confirm the practical applicability of the approximate proximal operator and are consistent with the theory connecting the approximation accuracy with the scaling parameter.
\end{enumerate}

\section{Related Work} \label{section: background}

The TV regularizer promotes sparse image gradients, producing reconstructed images that are approximately piecewise constant. 
TV regularization has been demonstrated to be effective in various inverse problems, including image deconvolution, diffraction ultrasound tomography, compressed sensing, and optical tomography~\cite{Bronstein.etal2002, Afonso.etal2010, Candes.etal2006, Lustig.etal2007, Louchet.Moisan2008, Oliveira.etal2009, Kamilov.etal2015b, Qu.etal2020, J.Liuetal2018, Guo.Chen2021, Kong.etal2022, Chen.etal2020, Takeyama.etal2017, Ma.etal2022, Diwakar.etal2024, Chambolle.Pock2016, Luo.etal2018, kadu.etal2020}. 
APGM~\cite{Beck.Teboulle2009} is often employed to solve image reconstruction problems with the TV regularizer~\cite{Beck.Teboulle2009a}. However, APGM requires additional sub-iterations to compute the proximal operator of the TV regularizer, which can slow down the overall convergence. 
Other approaches such as Primal-Dual Hybrid Gradient Algorithm (PDHG)~\cite{Chambolle.Pock2011} and ADMM~\cite{Afonso.etal2010} can be used to avoid these sub-iterations; but they have their own drawbacks such as increased memory consumption to represent the image gradients.

A fast non-iterative algorithm for solving the 1D TV proximal operator has been proposed~\cite{Condat2013}, but it does not generalize for signals with more than one dimension, such as images.
To generalize to an arbitrary number of dimensions, a closed-form approximation of the proximal operator for TV was proposed, eliminating the need for sub-iterations~\cite{Kamilov.etal2014a, Kamilov2016, Kamilov2016a, Kamilov2017}. 
The convergence of APGM with this approximation for anisotropic TV was investigated in~\cite{Kamilov2017} and for isotropic TV was investigated in~\cite{Kamilov2016a}. These approximations share a close relationship with image reconstruction techniques using wavelet domain regularization~\cite{Mallat2009}. They can also be viewed as the application of the \emph{proximal average}~\cite{Yu2013} approximation for the proximal operator of a sum of multiple nonsmooth functions. Cycle spinning is another closely related concept originally introduced for denoising~\cite{Coifman.Donoho1995}, later refined~\cite{Fletcher.etal2002}, and applied in various imaging inverse problems, including image restoration, MRI reconstruction, 3D CT reconstruction, and computed tomography~\cite{Figueiredo.Nowak2003, Vonesch.Unser2008, Vonesch.Unser2009, Guerquin-Kern.etal2011, Ramani.Fessler2012, Kamilov.etal2012, Borisch.etal2023, Ong.etal2018, Teyfouri.etal2021}.

The convergence of proximal algorithms has been extensively investigated~\cite{Parikh2014}. Our work is more related to the family of \emph{inexact} proximal algorithms that rely on approximations of proximal operators or gradient calculations~\cite{Rockafellar1976, Guler1992, Aujol.Dossal2015, dAspremont2008, Devolder.etal2013, Salzo.Villa2012, Villa2013, Schmidt.etal2011, Nesterov2023, Bello-Cruz2020, Bertsekas2011}.
In particular, our theoretical analysis shows that the approximation in~\cite{Kamilov.etal2014a, Kamilov2016, Kamilov2016a, Kamilov2017} leads to an inexact TV proximal operator whose accuracy depends on its scaling parameter, thus enabling its use within existing proximal methods.

\begin{figure*}[t] %
\begin{minipage}[t]{0.48\textwidth}
    \begin{algorithm}[H]%
    \caption{APGM with $\approxTV$}
         \begin{algorithmic}[1] \label{alg: approx APGM}
         \renewcommand{\algorithmicrequire}{\textbf{input:}}
         \REQUIRE $g$, $\xbm_0, \zbm_0, \sbm_0
         \in \R^{\sizeConst}$, $\gamma>0$, $\lambda>0$, and $\{q_k\}_{k\in \N}$
            \FOR {$k = 1,2,...$}
            \STATE $\zbm^k= \sbm^{k-1} - \gamma \nabla g(\sbm^{k-1})$
            \STATE $\xbm^k= \Scal_{\gamma \lambda}(\zbm^{k})$
            \STATE $\sbm^k = \xbm^k + ((q_{k-1} - 1)/q_k)(\xbm^k- \xbm^{k-1})$
            \ENDFOR
         \end{algorithmic}
     \end{algorithm}
\end{minipage}
\hfill
\begin{minipage}[t]{0.48\textwidth}
    \begin{algorithm}[H]
    \caption{ADMM with $\approxTV$}
         \begin{algorithmic}[1] \label{alg: approx ADMM}
         \renewcommand{\algorithmicrequire}{\textbf{input:}}
         \REQUIRE $g$, $\xbm_0, \zbm_0, \sbm_0
         \in \R^{\sizeConst}$, $\gamma>0$, $\lambda>0$
            \FOR {$k = 1,2,...$}
            \STATE $\zbm^k= \prox_{\gamma g}(\xbm^{k-1} - \sbm^{k-1})$
            \STATE $\xbm^k= \Scal_{\gamma \lambda}(\zbm^{k-1} + \sbm^{k-1})$
            \STATE $\sbm^k = \sbm^{k-1} + \xbm^{k} - \zbm^{k}$
            \ENDFOR
         \end{algorithmic}
     \end{algorithm}
\end{minipage}
\end{figure*}

\section{Approximating Total Variation} \label{section: approx tv}

\subsection{Method} \label{section: method}

In this section, we present the approximation of the TV proximal operator proposed in~\cite{Kamilov2017, Kamilov2016a}.
The motivation for this approximation is the reduction of the computational complexity of calculating the TV proximal operator from $O(knd)$, where $k$ is the number of sub-iterations, to $O(nd)$.
Consider the linear mapping $\Wbm: \R^{\sizeConst} \rightarrow \R^{2  \sizeConst \dimConst}$ defined as
\begin{align*}
    \Wbm := \frac{1}{2 \sqrt{d}} \begin{bmatrix}
            \Mbm \vspace{3pt} \\
            \Dbm
    \end{bmatrix} , \qquad \begin{aligned}
        \Mbm & := \big[\Mbm_1^{\Tsf} ~ \cdots ~\Mbm_d^{\Tsf} \big]^\Tsf, \\ 
        \Dbm & := \big[\Dbm_1^{\Tsf} ~ \cdots ~ \Dbm_d^{\Tsf} \big]^\Tsf,
    \end{aligned} 
\end{align*}
where $\Mbm_j:\ \R^\sizeConst \rightarrow \R^{\sizeConst  }$  for $j \in \{1, ..., \dimConst\}$ represents an averaging operator along the $j^{\text{th}}$ dimension, and $\Dbm_j:\ \R^{\sizeConst} \rightarrow \R^{\sizeConst }$  for $j \in \{1, ..., \dimConst\}$  denotes a discrete gradient operator along the $j^{\text{th}}$ dimension. 
$\Mbm_j$ and $\Dbm_j$ are convolutions along the $j^{\text{th}}$ dimension of the signal with the kernels
\begin{align*}
    \mbm = \begin{bmatrix}
            1 \vspace{3pt} \\
            1
    \end{bmatrix}, \qquad \dbm =\begin{bmatrix}
            1 \vspace{3pt} \\
            -1
    \end{bmatrix},
\end{align*}
respectively.
In this work, we assume periodic boundary conditions for both $\Mbm$ and $\Dbm$. Note that this is not essential for the method or the results in this paper: aperiodic boundary conditions are perfectly acceptable.
The mapping $\Wbm$  can be interpreted as a union of scaled and shifted first-level discrete Haar wavelet and scaling functions across all $\dimConst$ dimensions~\cite{Kamilov2017}.
It is important to note that this operator satisfies  $\Wbm^{\Tsf} \Wbm = \bm{I}$, but due to redundancy in the union of scaled orthogonal transforms, $\Wbm \Wbm^{\Tsf}$ is over-determined and so $\Wbm \Wbm^{\Tsf} \neq \bm{I}$~\cite{Elad.etal2007} (Section 3, first paragraph).

The approximate TV proximal operator is defined as
\begin{align}\label{eq: operator def}
    \approxTV \defn \Wbm^{\Tsf} \shrinkageFunction[][\tau \hspace{0.15em} \shrinkStepConst]{\Wbm \zbm},
\end{align}
where $\Tcal_{\lambda}$ is an operator with the parameter $\lambda \defn \tau \hspace{0.15em} \shrinkStepConst >0$, corresponding to the application of component-wise shrinkage functions on the scaled difference components of $\Wbm \zbm$.
    For the anisotropic operator $\Tcal^{\hspace{0.15em} a}_{\tau 2 \sqrt{d}}$, the shrinkage function
\begin{equation} \label{eq: aniso soft threshold}
    \max( | u | - \lambda, 0 ) \frac{u}{ | u | },
\end{equation} 
is applied to each difference component; that is, $u = [\Wbm \zbm]^{\text{dif}}_i \in \R$ for $i \in \{1,2,...,\sizeConst\dimConst\}$. 
For the isotropic operator $\Tcal^{\hspace{0.15em} i}_{\tau 2 \sqrt{d}}$, the shrinkage function
\begin{equation} \label{eq: iso soft threshold}
    \max( \| \ubm \|_2 - \lambda, 0 ) \frac{\ubm}{ \| \ubm \|_2 },
\end{equation}
is applied to each finite difference vector corresponding to each pixel in $\zbm$; that is, $\ubm = [\Wbm \zbm]^{\text{dif}}_{G_i} \in \R^{\dimConst}$ for $i \in \{ 1,2, \cdots, n\}$.
Each $G_i$ contains the $\dimConst$ indices in $[\Wbm \zbm]^{\text{dif}}$ corresponding to the components of the finite difference vector at pixel $i$.

Note that the shrinkage parameter $\tau$ must be scaled by  $\shrinkStepConst$ due to the scaling of $\Wbm$. 
 For simplicity, we omit the superscripts  $\text{a}$  and $\text{i}$ 
  when we do not need to distinguish between 
 the anisotropic and isotropic thresholding functions, and the corresponding TV functions.
 The following results apply to both versions.

\subsection{Theoretical Analysis} \label{section: analysis}

We now present the theoretical analysis $\approxTV$ as an approximation of the TV proximal operator.
The detailed proofs are provided in the Appendix~\ref{section: appendix}. 

\begin{proposition} \label{prop: proximal of another function}
    For $\tau > 0$, there exists a proper, closed, and convex function
    $\approxh$
    such that $\approxTV = \prox_{\approxh}(\zbm)$ for all $\zbm \in \R^{\sizeConst}$.
    
\end{proposition}

\begin{figure*}[t]
 \centering
 \includegraphics[scale=0.9]{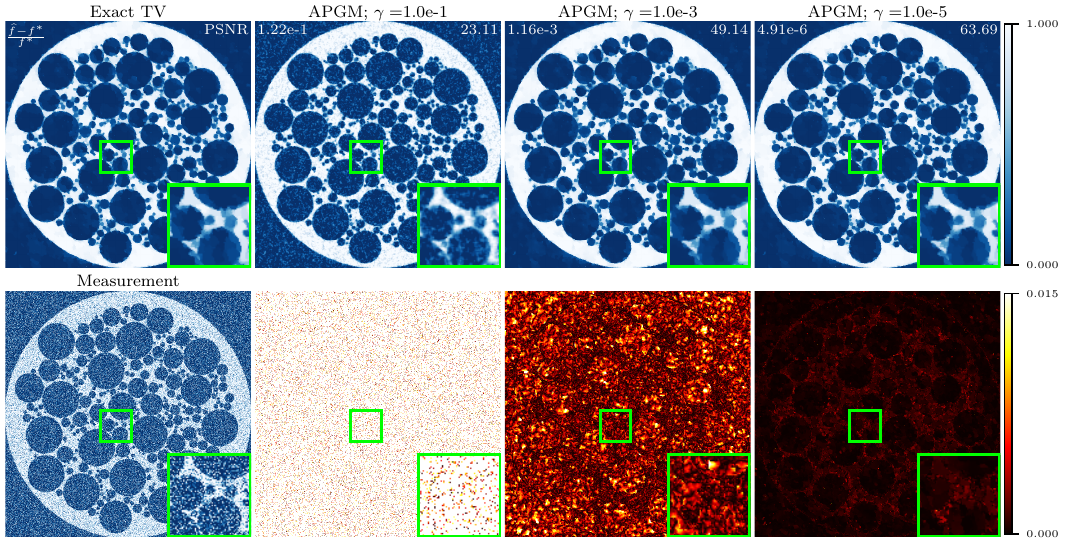}
 \caption{
    Effect of $\tau$ on image denoising performance when using $\approxTV$ in APGM compared to the exact TV reconstruction with regularization $\lambda = 0.5$.
    As written in Algorithm~\ref{alg: approx APGM}, $\tau = \gamma \lambda$, where $\gamma$ is the step-size.
    The top-left shows the relative cost of the approximate TV reconstruction to the exact TV reconstruction.
    The PSNR and difference images are relative to the exact TV reconstruction.
    Following Proposition~\ref{prop: main prop}, a smaller $\tau$ results in a smaller error.
 }
\label{fig: denoise comparison}
\end{figure*}

Proposition~\ref{prop: proximal of another function} establishes the existence of a function $\approxh$ for which $\approxTV$ is the proximal operator.
Although the explicit form of  $\approxh$ is unknown, its existence implies that $\approxTV$ can be incorporated into any proximal optimization method with the expectation of algorithm convergence. 

We will next demonstrate that $\approxTV$ provides a close approximation of the proximal operator of TV.
First, we show that it is related to the following smoothed versions of $h$.
In the anisotropic case, define
\begin{equation} \label{eq: huberized anisotropic}
   \widetilde{h}^a(\zbm) = \sum_{i=1}^{nd} \phi_{\tau 4 d}( [\Dbm \zbm ]_i),
\end{equation}
where $\phi_{\theta}: \R \rightarrow \R$ is defined as
\begin{equation} \label{eq: huberized anisotropic helper}
    \phi_{\theta} (y) = \begin{cases}
        \frac{y^2}{2 \theta} & \text{for } |y| \leq \theta \\
        |y | - \frac{\theta}{2} & \text{for } |y| > \theta
    \end{cases}.
\end{equation}
For the isotropic case, define
\begin{equation} \label{eq: huberized isotropic}
   \widetilde{h}^i(\zbm) = \sum_{i=1}^{n} \psi_{\tau 4 d}( [\Dbm \zbm ]_{G_i}),
\end{equation}
where for $\ybm \in \R^d$ we define
\begin{equation} \label{eq: huberized isotropic helper}
    \psi_{\theta} (\ybm) = \begin{cases} 
        \frac{\| \ybm \|^2}{2 \theta} & \text{for } \| \ybm \| \leq \theta \\
        \| \ybm \| - \frac{\theta}{2} & \text{for } \| \ybm \| > \theta 
    \end{cases}.
\end{equation}
Note that each $G_i$ is the collection of $d$ indices of the vector $\Dbm \zbm$ corresponding to the finite difference vector at position $i$ of $\zbm$.
Related approximations have previously appeared in the TV literature~\cite{Chambolle.Pock2011, selesnick2017, Hong.etal2020, xu.etal2018, Pock.etal2010, Pagliari.etal2022}, based on the Huber function~\cite{huber1964} 
or Moreau-Yoshida regularization~\cite{Moreau.1965, yosida1963}.
\begin{proposition}\label{prop: equiv to huber}
    Let $\widetilde{h}^a(\zbm)$ and $\widetilde{h}^i(\zbm)$ denote the Huber variants of anisotropic TV and isotropic TV given in Eq.~\eqref{eq: huberized anisotropic} and Eq.~\eqref{eq: huberized isotropic}, respectively.
    Then,
    \[\approxTV[a]=\zbm - \tau \nabla \widetilde{h}^a(\zbm) \quad \text{and} \quad \approxTV[i]=\zbm - \tau \nabla \widetilde{h}^i(\zbm).
    \]
\end{proposition} 
Proposition \ref{prop: equiv to huber} states that the operator $\approxTV$ is equivalent to a gradient step with step-size $\tau$ on a smoothed version of TV.
A recent work \cite{Kowalski2024} demonstrates that running a single iteration of the dual forward backward algorithm to solve the TV proximal within a FPG algorithm results in an algorithm that converges to a solution of the TV minimized problem.
Proposition \ref{prop: equiv to huber} shows that $\approxTV$ is related; however, instead of running one iteration of the dual-FB algorithm, $\approxTV$ corresponds to one iteration of gradient descent on a smoothed version of TV.
To decrease 
Eq.~\eqref{eq: huberized anisotropic} or Eq.~\eqref{eq: huberized isotropic}, the step-size must be at most $1/L$, where $L$ is the Lipschitz constant of the function's gradient~\cite{boyd.vandenberghe.2004}~(Section 9.3).
For both Eq.~\eqref{eq: huberized anisotropic} and Eq.~\eqref{eq: huberized isotropic}, the Lipschitz constant is $1/{\tau}$, as verified in the Appendix, satisfying the condition.

To further describe the operator as an approximation of the true proximal operator for $\TV$, we introduce the concept of the subdifferential~\cite{Moreau.1965, Rockafellar.1970}.
\begin{definition} \label{def: subdifferential}
    The subdifferential of the convex function $f: \R^{\sizeConst} \rightarrow (-\infty, +\infty]$ at $\ybm \in \text{dom}(f)$ is defined to be
    \begin{align*}
        \partial f(\ybm) := \big\{ \gbm \in \R^{\sizeConst}: f(\xbm) \geq f(\ybm) + \gbm^{\Tsf} (\xbm - \ybm) , \hspace{0.15em} \forall \xbm \in \R^{\sizeConst}  \big\}.
    \end{align*}
\end{definition}
A relaxation of this definition is called the $\epsilon$-subdifferential~\cite{Brondsted.Rockafellar1965} and defined as
\begin{definition} \label{def: epsilon subdifferential}
    The $\epsilon$-subdifferential of the convex function $f: \R^n \rightarrow (-\infty, +\infty]$ at $\ybm \in \text{dom}(f)$ is defined to be
    \begin{align*}
        \partial_{\epsilon}&f(\ybm) 
        \defn \big\{ \gbm \in \R^{\sizeConst}: f(\xbm) \geq f(\ybm) + \gbm^{\Tsf} (\xbm - \ybm) - \epsilon, 
             \forall \xbm \in \R^{\sizeConst}  \big\}.
    \end{align*}
\end{definition}
\noindent By comparing Definition~\ref{def: epsilon subdifferential} to Definition~\ref{def: subdifferential}, it becomes evident  that the traditional subdifferential is a subset of the $\epsilon$-subdifferential.
The $\epsilon$-subdifferential represents the subdifferentials relaxed by $\epsilon$.

Substantial work has been done to define and analyze various notions of approximate proximal operators~\cite{Rockafellar1976, Schmidt.etal2011, Nesterov2023, Salzo.Villa2012}. 
In Proposition~\ref{prop: main prop}, we prove that the operator $\approxTV$ can be characterized using several of the common  definitions for an approximate proximal operator~\cite{Salzo.Villa2012}.

\begin{proposition} \label{prop: main prop}
    For all $\zbm \in \R^{\sizeConst}$, the operator $\approxTV$ satisfies
    \begin{enumerate}[label=(\alph*)]
        \item $\approxTV = \prox_{\tau \TV}(\zbm + \deltabm)$, where $\| \deltabm \|_2 \leq \tau \operatorBoundConst$, \label{prop: type 3}
        \vspace{0.5em}
        \item $\zbm - \approxTV \in \tau \partial_{\tau \operatorSubdiffConst} \TV(\approxTV)$, \label{prop: type 2}
    \end{enumerate}
    where $\epsilon_1$ and $\epsilon_2$ are constants.
\end{proposition}
Proposition \ref{prop: main prop} relies 
on notions of approximate proximal operators, commonly used in the literature. Specifically, Proposition\ref{prop: main prop}\,(b) aligns with  a ``type 3 approximation'' in the proximal operator with $\tau \epsilon_1$-precision, while Proposition \ref{prop: main prop}\,(c) corresponds to a ``type 2 approximation'' with $\tau \sqrt{2 \epsilon_2}$-precision~\cite{Salzo.Villa2012}.

The first part of Proposition~\ref{prop: main prop} demonstrates that the operator functions as a perturbed proximal operator of TV.  Specifically, the output of the operator is equivalent to applying the true proximal operator of TV to a perturbed input vector $\zbm$, modified by a vector $\deltabm$.
Notably, the norm of the perturbation $\|\deltabm\|_2$ is bounded and can be controlled by adjusting $\tau$. From the optimality condition for the proximal operator of true TV regularizer $\prox_{ \tau \TV }(\zbm)$, there is the following well known relationship
\begin{align*}
    \zbm - \prox_{\tau \TV}(\zbm) \in \tau \partial \hspace{0.15em} \TV( \prox_{\tau \TV}(\zbm) ) \;.
\end{align*}
The second part of Proposition~\ref{prop: main prop} establishes an analogous connection for the closed-form approximate operator, characterized using the  $\epsilon$-subdifferential from Definition~\ref{def: epsilon subdifferential}, where $\epsilon$ is controlled by the proximal scaling parameter $\tau$.

Overall, Proposition~\ref{prop: main prop} implies that the $\approxh$ from Proposition~\ref{prop: proximal of another function} can be seen as an approximation of the true TV function $\TV$ with controllable accuracy. Notably, by reducing  $\tau$, one can achieve an arbitrarily close approximation to the true TV proximal operator. Therefore, our theoretical result supports the use of  $\approxTV$ as an approximation of $\prox_{\tau \TV}(\zbm)$. 
As the accuracy of the operator is controlled by $\tau$, it might initially appear that it would only be effective for problems with a very small regularization parameter  $\lambda$  for TV. However, this is not the case when $\approxTV$ is used within iterative algorithms such as APGM and ADMM. As shown in Algorithms~\ref{alg: approx APGM} and~\ref{alg: approx ADMM}, the shrinkage parameter for $\approxTV$  is set by  $\tau = \gamma \lambda$. Thus, one can enhance accuracy by adjusting  $\gamma $ rather than altering the TV regularization parameter  $\lambda$ . For APGM,  $\gamma$  represents the step size, while for ADMM, $ \gamma$  serves as the penalty parameter in the augmented Lagrangian.

 \begin{figure*}[t]
 \centering
 \includegraphics[scale=0.9]{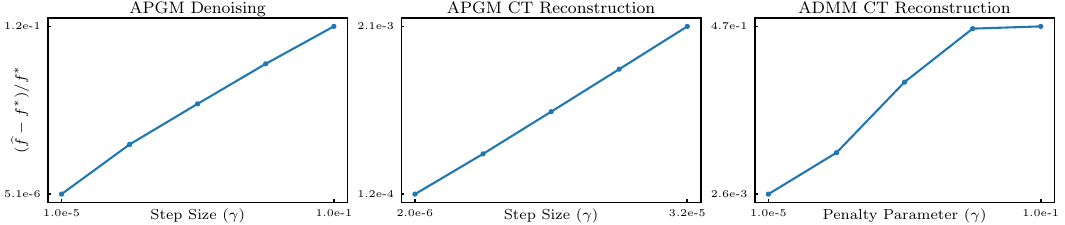}
 \caption{
    Effect of $\tau$ on the approximate reconstruction 
    performance $\widehat{f}$ using the $\approxTV$
    relative to the exact TV reconstruction $f^*$.
    $\approxTV$ is tested within the APGM and ADMM algorithms.
    For image denoising $\lambda=0.5$ and for CT reconstruction $\lambda=5$.
    In both cases, as written in Algorithms~\ref{alg: approx APGM} and~\ref{alg: approx ADMM}, $\tau = \gamma \lambda$, where $\gamma$ is the step-size in APGM and the penalty parameter in ADMM.
    As expected from Proposition~\ref{prop: main prop}, the smaller the $\tau$, the smaller the error.
 }
\label{fig: pattern compare}
\end{figure*}
\begin{table*}[t]
    \fontsize{5}{6}\selectfont
    \caption{The relative cost, PSNR with respect to the exact TV solution and groundtruth, and acceleration achieved by using $\approxTV$ (measured as the ratio of the number of iterations when using $\approxTV$ to total number of sub-iterations when using the exact TV proximal)
    for various regularization ($\lambda$) and step-size/penalty ($\gamma$) parameters for image denoising and CT reconstruction.
    For the APGM and ADMM algorithms,  $\tau=\gamma \lambda$, where $\gamma$ determines the accuracy of $\approxTV$.
    The metrics have been averaged across the 10 images in the dataset.
    Note that $L$ is the Lipschitz constant of $\nabla g$.\label{table: average performance}}
    \renewcommand{\arraystretch}{1.25} %
    \centering
    {\fontfamily{ptm}\selectfont %
    \vspace{1mm}
    \begin{tabular}{|c|c c c c|c c c c|c c c c|}
    \multicolumn{13}{ c }{\rule{0pt}{3ex}\textbf{APGM Image Denoising} } \\
    
    \hline
    \multicolumn{1}{|c|}{} & \multicolumn{4}{c|}{\underline{$\lambda = 0.25$}} & \multicolumn{4}{c|}{\underline{$\lambda = 0.5$}} & \multicolumn{4}{c|}{\underline{$\lambda = 1$}} \\
    \textbf{Step Size ($\gamma$)} & \textbf{Rel. Err.} & \textbf{PSNR (TV)} & \textbf{PSNR (GT)} & \textbf{Accel} & \textbf{Rel. Err.} & \textbf{PSNR (TV)} & \textbf{PSNR (GT)} & \textbf{Accel} & \textbf{Rel. Err.} & \textbf{PSNR (TV)} & \textbf{PSNR (GT)} & \textbf{Accel} \\
    \hline
    1e-1 & 1.751e-02 & 31.97 & 16.44 & $2.32\times$ & 1.214e-01 & 23.14 & 17.57 & $9.58\times$ & 3.875e-01 & 19.65 & 17.38 & $19.86\times$ \\
    1e-2 & 1.364e-03 & 48.02 & 16.72 & $0.41\times$ & 1.281e-02 & 34.52 & 18.57 & $1.61\times$ & 5.467e-02 & 25.82 & 15.68 & $2.67\times$ \\
    1e-3 & 1.152e-04 & 65.90 & 16.72 & $0.09\times$ & 1.157e-03 & 49.16 & 18.75 & $0.38\times$ & 5.945e-03 & 36.31 & 15.96 & $0.67\times$ \\
    \hline

    \multicolumn{13}{ c }{\rule{0pt}{3ex}\textbf{APGM CT Reconstruction} } \\

    \hline
    \multicolumn{1}{|c|}{} & \multicolumn{4}{c|}{\underline{$\lambda = 2.5$}} & \multicolumn{4}{c|}{\underline{$\lambda = 5$}} & \multicolumn{4}{c|}{\underline{$\lambda = 10$}} \\
    \textbf{Step Size ($\gamma$)} & \textbf{Rel. Err.} & \textbf{PSNR (TV)} & \textbf{PSNR (GT)} & \textbf{Accel} & \textbf{Rel. Err.} & \textbf{PSNR (TV)} & \textbf{PSNR (GT)} & \textbf{Accel} & \textbf{Rel. Err.} & \textbf{PSNR (TV)} & \textbf{PSNR (GT)} & \textbf{Accel} \\
    \hline
    1/(1$L$) & 9.473e-04 & 57.28 & 21.07 & $49.50\times$ & 2.069e-03 & 52.08 & 20.86 & $49.21\times$ & 4.588e-03 & 46.48 & 20.41 & $49.71\times$ \\
    1/(2$L$) & 4.609e-04 & 59.46 & 21.06 & $39.99\times$ & 1.002e-03 & 56.87 & 20.88 & $40.20\times$ & 2.208e-03 & 51.84 & 20.48 & $39.21\times$ \\
    1/(4$L$) & 2.264e-04 & 62.14 & 21.07 & $33.51\times$ & 4.874e-04 & 60.42 & 20.88 & $32.99\times$ & 1.068e-03 & 56.79 & 20.51 & $31.09\times$ \\
    \hline

    \multicolumn{13}{ c }{\rule{0pt}{3ex}\textbf{ADMM CT Reconstruction}} \\
    \hline
    
    \multicolumn{1}{|c|}{} & \multicolumn{4}{c|}{\underline{$\lambda = 2.5$}} & \multicolumn{4}{c|}{\underline{$\lambda = 5$}} & \multicolumn{4}{c|}{\underline{$\lambda = 10$}} \\
    \textbf{Pen. Par. ($\gamma$)} & \textbf{Rel. Err.} & \textbf{PSNR (TV)} & \textbf{PSNR (GT)} & \textbf{Accel} & \textbf{Rel. Err.} & \textbf{PSNR (TV)} & \textbf{PSNR (GT)} & \textbf{Accel} & \textbf{Rel. Err.} & \textbf{PSNR (TV)} & \textbf{PSNR (GT)} & \textbf{Accel} \\
    \hline
    1e-2 & 3.545e-01 & 18.49 & 15.35 & $260.90\times$ & 4.341e-01 & 18.10 & 15.13 & $157.57\times$ & 4.680e-01 & 18.28 & 15.13 & $89.96\times$ \\
    1e-3 & 3.767e-02 & 31.63 & 19.82 & $62.75\times$ & 8.291e-02 & 26.42 & 18.59 & $69.34\times$ & 1.788e-01 & 21.91 & 16.82 & $69.53\times$ \\
    1e-4 & 4.850e-03 & 40.42 & 20.59 & $11.43\times$ & 9.284e-03 & 38.55 & 20.37 & $10.24\times$ & 1.989e-02 & 35.13 & 19.82 & $9.90\times$ \\
    \hline
    \end{tabular}
    }
\end{table*}

\section{Numerical Validation}

\label{section: experiments}
In this section, we numerically validate Proposition~\ref{prop: main prop} by comparing the reconstruction performance of iterative proximal algorithms when $\approxTV$ replaces the true proximal operator for isotropic total variation.
For experimental validation of the efficiency and timing of these operators, we refer the reader to \cite{Kamilov2016, Kamilov2017, Kamilov2016a}.
Specifically, we run APGM and ADMM on a simulated low-angle computed tomography (CT) image reconstruction problem.
The APGM and ADMM algorithms using $\approxTV$ instead of the true proximal operator are shown in Algorithms~\ref{alg: approx APGM} and~\ref{alg: approx ADMM}, respectively.
The accuracy of $\approxTV$ can be directly controlled by $\gamma$ since $\tau = \gamma \lambda$, where  $\gamma$ is the step size and penalty parameter, for APGM and ADMM algorithms, respectively.
To calculate the true proximal operator, we use the iterative Fast Projected Gradient (FPG) algorithm~\cite{Beck.Teboulle2009a}.
Additionally, we investigate the performance of the closed-form approximate operator relative to true proximal operator of TV for an image denoising task. 
The numerical evaluations are performed using 10 piecewise-smooth foam phantoms generated by the XDesign Python package~\cite{ching_gursoy.2017} and 
the implementations\footnote{See classes \href{https://scico.readthedocs.io/en/stable/\_autosummary/scico.functional.html\#scico.functional.AnisotropicTVNorm}{AnisotropicTVNorm} and \href{https://scico.readthedocs.io/en/stable/\_autosummary/scico.functional.html\#scico.functional.IsotropicTVNorm}{IsotropicTVNorm}.} of $\approxTV$ in the SCICO Python package~\cite{Balke.etal2022}.
All algorithms are run until the stopping criteria 
\begin{align} \label{stopping criteria}
    \frac{\|\xbm^t - \xbm^{t-1}\|_2}{\|\xbm^{t-1}\|_2} \leq 5 \times 10^{-6}
\end{align}
is met.

\subsection{Image Denoising} \label{section: image denoising}
In the case of image denoising, the forward model is $\Abm = \Ibm$, which results in the  data-fidelity term $g(\xbm) = \frac{1}{2} \| \ybm - \xbm\|_2^2$. 
We report the denoising results for the regularization parameters $0.25$, $0.5$, and $1$ in Table~\ref{table: average performance}.
Since FPG solves the denoising problem, it is directly used to compute the exact TV regularized denoised image.
The APGM algorithm with $\approxTV$ is used to compute the approximate TV regularized images.

\subsection{Computed Tomography} \label{section: computed tomography}

We validate the application of $\approxTV$ for the limited-angle computed tomography (CT) reconstruction. 
The data fidelity term in Eq.~\eqref{eq: main optimization problem} is 
$g(\xbm) = \frac{1}{2} \| \Abm \xbm - \ybm \|_2^2, $
where $\Abm$ represents the CT imaging forward operator using only $45$ angles, resulting in an ill-posed inverse problem.
We show the experimental results for the regularization parameters $\lambda$ corresponding to $2.5$, $5$, and $10$, using both APGM and ADMM. 
ADMM is implemented to solve the constrained optimization problem 
\[
\argmin_{\xbm, \zbm} g(\xbm) + \lambda h(\zbm) \quad \text{s.t.} \;\; \zbm = \xbm
\]
in order to test the impact of $\approxTV$.
The exact TV reconstruction is obtained by running 50 sub-iterations of FPG to calculate the exact proximal within both algorithms.

\subsection{Discussion} \label{section: discussion}
According to Proposition~\ref{prop: proximal of another function}, $\approxTV$ serves as the proximal operator for some convex function. 
Consequently, proximal-based reconstruction algorithms using this operator are guaranteed to converge.
Activation of the termination criteria in~\ref{stopping criteria} occurs for all the experiments, validating the convergence of the algorithms under $\approxTV$.

 \begin{figure*}[t]
 \centering
 \includegraphics[scale=0.9]{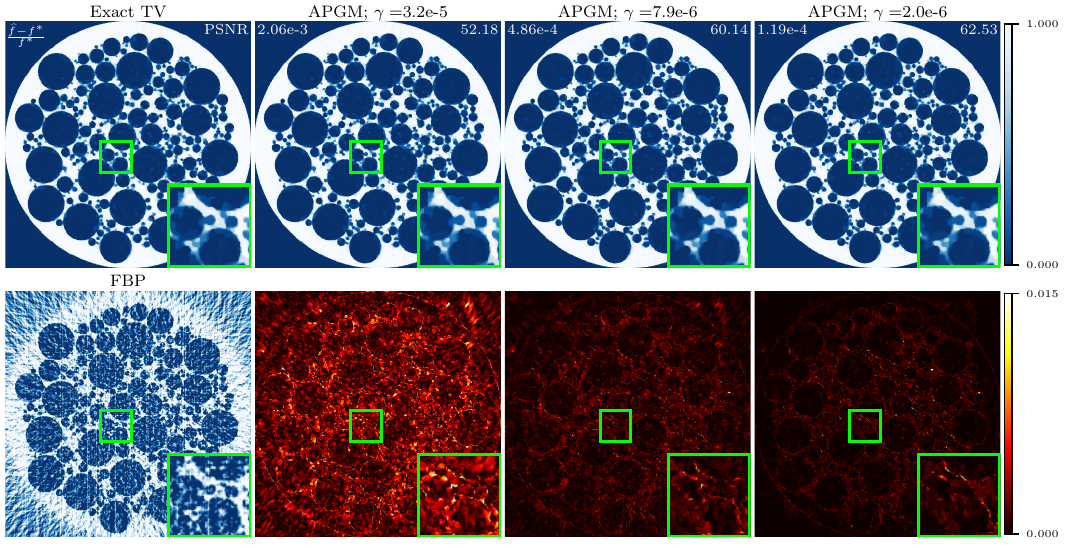}
 \caption{ 
    Effect of $\tau$ on limited angle Computed Tomography (CT) reconstruction using $\approxTV$ in APGM compared to the exact TV reconstruction with regularization $\lambda=5$.
    As written in Algorithm~\ref{alg: approx APGM}, $\tau = \gamma \lambda$, where $\gamma$ is the step-size.
    The top-left shows relative cost of the approximate reconstruction to the exact reconstruction.
    The PSNR and difference images are relative to the exact reconstruction.
    Following Proposition~\ref{prop: main prop}, a smaller $\tau$ results in a smaller error.
}
\label{fig: apgm reconstruction}
\end{figure*}

The statements in Proposition~\ref{prop: equiv to huber} and Proposition~\ref{prop: main prop} imply that a smaller $\tau$ in $\approxTV$ provides a better approximation with respect to the true TV reconstruction.
This relationship is experimentally validated in Table~\ref{table: average performance} across all experiments: a smaller  $\tau$ results in a smaller relative error in the approximate reconstruction loss $\widehat{f}$, compared to the true TV reconstruction loss $f^*$. 
Relative error, denoted as Rel. Err. in Table~\ref{table: average performance}, is defined to be
$$\frac{\widehat{f} - f^{*}}{f^{*}}.$$
Similarly, the PSNR with respect to the true TV reconstruction increases as $\tau$ decreases, suggesting $\approxTV$ provides a better approximation of the TV proximal operator. 
Table~\ref{table: average performance} also displays the trade-off between approximation accuracy and acceleration when using $\approxTV$ compared to $50$ iterations of FPG for the exact proximal.
``Accel'' is computed as the number of iteration using $\approxTV$ divided by total number of sub-iterations when using FPG for the exact proximal.
Figure~\ref{fig: pattern compare} illustrates the dependence of cost accuracy on parameter $\gamma$ for image denoising with $\lambda=0.5$ and CT image reconstruction with $\lambda=5$.
Figures~\ref{fig: denoise comparison},  ~\ref{fig: apgm reconstruction}, and~\ref{fig: admm reconstruction} illustrate visual comparison for image denoising and CT reconstruction using true TV proximal and $\approxTV$ operator with various $\gamma$. 
The visual results indicate that $\approxTV$ promotes piecewise smooth images, as is expected from an approximate TV proximal operator.
Both visual and numerical results from the experiments support the theoretical finding  presented in Section~\ref{section: analysis},  demonstrating that $\approxTV$ can approximate the TV proximal and be effectively integrated into proximal-based iterative algorithms to achieve approximate TV-regularized reconstructions.

 \begin{figure*}[t]
 \centering
 \includegraphics[scale=0.9]{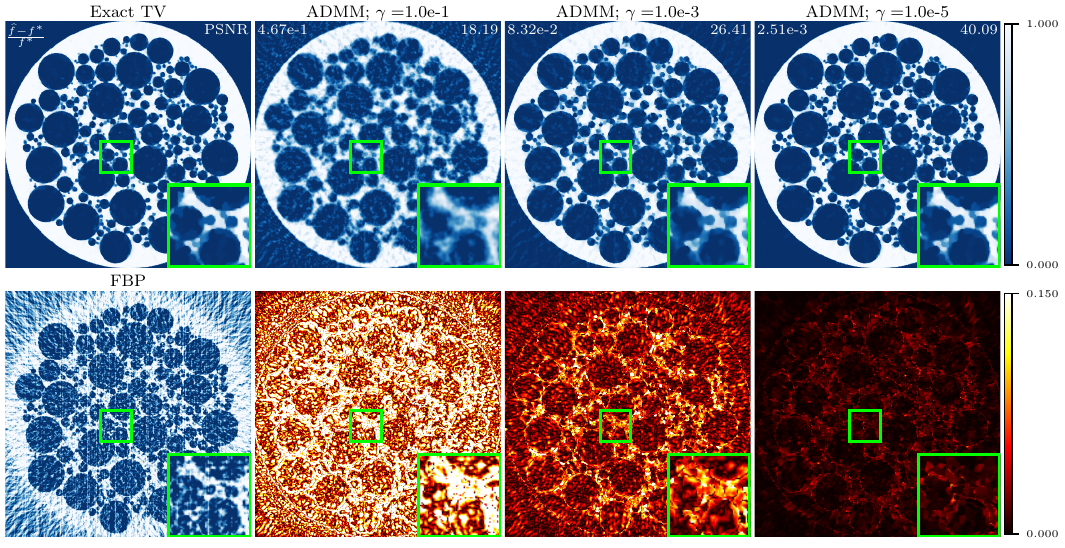}
 \caption{
    Effect of $\tau$ on limited angle Computed Tomography (CT) reconstruction using $\approxTV$ in ADMM compared to  exact reconstruction with regularization $\lambda=5$.
    As written in Algorithm~\ref{alg: approx ADMM}, $\tau = \gamma \lambda$, where $\gamma$ is the penalty parameter.
    The top-left shows relative cost of the approximate reconstruction to the exact reconstruction.
    The PSNR and difference images are relative to exact reconstruction.
    Following Proposition~\ref{prop: main prop}, a smaller $\tau$ results in a smaller error.
 }
\label{fig: admm reconstruction}
\end{figure*}
\section{Conclusion} \label{section: conclusion}
We theoretically analyzed a closed-form operator that approximates the proximal operator for both anisotropic and isotropic total variation. While these operators have been previously explored in the context of optimization algorithms, this work focuses on the operator itself. Specifically, we demonstrated that the operator consistently reduces the TV function of the input signal and that its error, relative to the true TV proximal operator, can be characterized using well-established notions of approximate proximal operators.
Through this analysis, we provide a foundation for employing $\approxTV$ within any proximal-based reconstruction method to address TV-regularized imaging inverse problems. Finally, we supported our theoretical findings with several experiments.

\section*{Acknowledgment}
The authors thank the anonymous reviewers for their
careful reading and comments that helped improve the final version of this manuscript.

\section{Appendix}\label{section: appendix}
In this section, we will prove Propositions ~\ref{prop: proximal of another function},~\ref{prop: equiv to huber}, and ~\ref{prop: main prop}.
The linear map $\Wbm: R^n \rightarrow \R^{2nd}$ is defined as
\begin{align*}
    \Wbm := \frac{1}{2 \sqrt{d}} \begin{bmatrix}
            \Mbm \vspace{3pt} \\
            \Dbm
    \end{bmatrix} , \qquad \begin{aligned}
        \Mbm & := \big[\Mbm_1^{\Tsf} ~ \cdots ~\Mbm_d^{\Tsf} \big]^\Tsf, \\ 
        \Dbm & := \big[\Dbm_1^{\Tsf} ~ \cdots ~ \Dbm_d^{\Tsf} \big]^\Tsf,
    \end{aligned} 
\end{align*}
where $\Mbm$ is the averaging map and $\Dbm$ is the finite difference map.
Therefore, for any vector $\ubm \in \R^{2nd}$ half of the  components correspond to the averaging map and the remaining components correspond to the finite difference map.
We use the notation $\ubm^{\text{dif}}$ to correspond to the difference coefficients.

Using this notation, we define the function $\helperh: \R^{2 \sizeConst \dimConst} \rightarrow \R$, where  $\helperh$ has an anisotropic and isotropic form, $\helperh[\text{a}]$ and $\helperh[\text{i}]$, respectively.
Let $G_1, ..., G_n$ be a partition of $\{1,2, ..., nd\}$ such that each $G_i$ is the set of indices corresponding to the finite difference of pixel $i$.
For $\ubm \in \R^{2 \sizeConst \dimConst}$,
\begin{align} %
    \helperh[\text{a}](\ubm) & := 2\sqrt{\dimConst}\| \ubm^{\text{dif}} \|_{1,1} = 2\sqrt{\dimConst} \sum_{i=1}^{\sizeConst} \| [\ubm^{\text{dif}}]_{G_i} \|_1 \label{eq: aniso helper h definition}  \\  
    \helperh[\text{i}](\ubm) & := 2\sqrt{\dimConst}\| \ubm^{\text{dif}} \|_{2,1}  = 2\sqrt{\dimConst} \sum_{i=1}^{\sizeConst} \| [\ubm^{\text{dif}}]_{G_i} \|_2 \;, \label{eq: iso helper h definition}
\end{align}
where $\| \ubm^{\text{dif}} \|_{p, 1}$ is a group norm, with $p$  either $1$ or $2$,   to anisotropic and isotropic TV, respectively.
The grouping is based on the  $\dimConst$  elements that represent the discrete differences at a specific location for each of the  $\dimConst$ dimensions, with the  $p$-norm applied within each group. 
Each $[\ubm^{\text{dif}}]_{G_i} \in \R^{d}$ is the vector 
at the indices given by the set $G_i$.
We  use
\begin{align}
    \helperh(\ubm) := 2\sqrt{\dimConst}\| \ubm^{\text{dif}} \|_{p, 1} \label{eq: helper h definition}
\end{align}
to refer to both the anisotropic and isotropic versions. Thus the results in Proposition~\ref{prop: proximal of another function} and Proposition~\ref{prop: main prop} hold for both anisotropic and isotropic versions.
We can also rewrite the TV definitions in Eq.~\eqref{eq: aniso TV def} and Eq.~\eqref{eq: iso TV def} in the general form
\begin{align}
    \TV(\zbm) = \| \Dbm \zbm \|_{p,1} \;. \label{eq: TV def}
\end{align}
The term $2\sqrt{\dimConst}$ in front of $\helperh$ is necessary so that the relationship 
\begin{align}
    \helperh(\Wbm \zbm) & = \TV(\zbm) \label{eq: helperh related to tv} 
\end{align}
holds.
Note that Eq.~\eqref{eq: helperh related to tv} 
implies that for $\zbm \in \R^n$, $\helperh$ evaluated at $\Wbm \zbm$ is equivalent to $\TV$ evaluated at $\zbm$; however, $\helperh$ is not the TV function itself since it is defined on the entire $\R^{2nd}$, and not $\R^n$.

Additionally, $\helperh$ is defined as the scaled version of two well-known norms that are proper, closed, and convex.
By using the notation $\Gamma^0$ to represent the set of proper, closed, and convex functions, we have $\helperh \in \Gamma^0(\R^{2nd})$.
Finally, define the linear subspace $\Ucal := \{\ubm | \ \ubm = \Wbm \Wbm^{\Tsf} \ubm \} \subset \R^{2 \sizeConst \dimConst}$.

\section*{Proof of Proposition~\ref{prop: proximal of another function}}

\begin{repproposition}{prop: proximal of another function}
    For $\tau > 0$, there exists $\approxh \in \Gamma^0(\R^{\sizeConst})$ such that $\approxTV = \prox_{\approxh}(\zbm)$ for all $\zbm \in \R^{\sizeConst}$.
\end{repproposition}
\begin{proof}
     To establish that $\approxTV$ is a proximal operator of some $\approxh \in \Gamma^0(\R^{\sizeConst})$, it is sufficient to show that 
    \begin{enumerate}
        \item There exists a closed and convex $\psi$ such that $\approxTV \in \partial \psi(\zbm),\ \forall \zbm \in \R^{\sizeConst}$.
        \item The operator $\approxTV$ is nonexpansive.
    \end{enumerate} 
    This equivalence was given in Corollary 10.c in~\cite{Moreau.1965}. 
    As defined in Eq.~\eqref{eq: helper h definition}, $\helperh$ is the scaled norm of the difference components of its input vector in $\R^{2nd}$.
    Therefore, the proximal operator of $\tau \helperh(\ubm)$
    can be written 
    using the thresholding function  defined in Eq.~\eqref{eq: aniso soft threshold} and Eq.~\eqref{eq: iso soft threshold}: 
    \begin{align}\label{eq: prox of hhat}
        \prox_{\tau \helperh}(\ubm) = \shrinkageFunction[][\tau \hspace{0.15em} \shrinkStepConst]{\ubm} \quad \forall \ubm\in \R^{2 \sizeConst \dimConst} \;. 
    \end{align} 
    By Corollary 10.c in~\cite{Moreau.1965}, there exists a closed and convex $\phi$ such that $\Tcal_{2 \tau \sqrt{\dimConst}}(\ubm) \in \partial \phi(\ubm)$.
    Plugging in $\ubm = \Wbm \zbm$ and multiplying both sides by $\Wbm^{\Tsf}$ gives us $\approxTV \in \Wbm^{\Tsf} \partial \phi(\Wbm \zbm)$.
    By the subdifferential chain rule, we know that 
    $\Wbm^{\Tsf} \partial \phi(\Wbm \zbm) \subset \partial [\phi \circ \Wbm](\zbm)$, where $\circ$ denotes composition of the functions.
    Therefore, we have
    \begin{align*}
        \approxTV \in \partial [\phi \circ \Wbm](\zbm) \;.
    \end{align*}
    Next, letting $\psi(\zbm) = [\phi \circ \Wbm](\zbm)$, we need to show $\psi$ is closed and convex.
    Since $\phi$ is closed, $\phi \circ \Wbm$ is also closed.
    Since $\phi$ is convex and $\Wbm$ is linear, $\phi \circ \Wbm$ is convex.
    We can show $\| \Wbm \|_2 = 1$ as follows
    \begin{align}
        \| \Wbm \|_2 
            & = \sqrt{\sigma_{\text{max}}(\Wbm^{\Tsf}\Wbm )}  \nonumber \\
            & = \sqrt{\sigma_{\text{max}}( \bm{I} )} = 1 \;, \label{eq: spectral norm W^T}
    \end{align}
    where $\sigma_{\text{max}}(\Ibm)$ is the largest eigenvalue of $\Ibm$.
    
    We now verify the nonexpansiveness of $\approxTV$.
    \begin{align*}
         & \| \Scal_{\tau}(\zbm_1) - \Scal_{\tau}(\zbm_2) \|_2 \\
         & \quad = \| \Wbm^{\Tsf} \shrinkageFunction[][\tau \hspace{0.15em} \shrinkStepConst]{\Wbm \zbm_1} - \Wbm^{\Tsf} \shrinkageFunction[][\tau \hspace{0.15em} \shrinkStepConst]{\Wbm \zbm_2} \|_2 \\
         & \quad  = \| \Wbm^{\Tsf} (\shrinkageFunction[][\tau \hspace{0.15em} \shrinkStepConst]{\Wbm \zbm_1} - \shrinkageFunction[][\tau \hspace{0.15em} \shrinkStepConst]{\Wbm \zbm_2}) \|_2 \\
         & \quad \leq \| \Wbm^{\Tsf} \|_2 \|\shrinkageFunction[][\tau \hspace{0.15em} \shrinkStepConst]{\Wbm \zbm_1} - \shrinkageFunction[][\tau \hspace{0.15em} \shrinkStepConst]{\Wbm \zbm_2} \|_2 \\
         & \quad = \|\shrinkageFunction[][\tau \hspace{0.15em} \shrinkStepConst]{\Wbm \zbm_1} - \shrinkageFunction[][\tau \hspace{0.15em} \shrinkStepConst]{\Wbm \zbm_2} \|_2 \\
         & \quad \leq \|\Wbm \zbm_1 - \Wbm \zbm_2 \|_2 \\
         & \quad \leq \|\Wbm \|_2 \|\zbm_1 - \zbm_2 \|_2 \\
         & \quad = \|\zbm_1 - \zbm_2 \|_2 \;, 
    \end{align*}
    Where the third and last equalities comes from Eq.~\eqref{eq: spectral norm W^T}. 
    The second inequality comes from the fact that soft-thresholding is a non-expansive operator.
    The above inequality is exactly the definition of a nonexpansive operator. 
\end{proof}

\section*{Proof of Proposition~\ref{prop: equiv to huber}}
\begin{repproposition}{prop: equiv to huber}
    Let $\widetilde{h}^a(\zbm)$ and $\widetilde{h}^i(\zbm)$ denote the Huber variants of anisotropic TV and isotropic TV given in Eq.~\eqref{eq: huberized anisotropic} and Eq.~\eqref{eq: huberized isotropic}, respectively.
    Then,
    \[\approxTV[a]=\zbm - \tau \nabla \widetilde{h}^a(\zbm) \quad \text{and} \quad \approxTV[i]=\zbm - \tau \nabla \widetilde{h}^i(\zbm).
    \]
\end{repproposition}
\begin{proof}
We will first consider anisotropic TV $\approxTV[a]$.
The soft-thresholding operator can be rewritten as 
\begin{equation*}
    \shrinkageFunction[\hspace{0.15em} a][\lambda]{u}  \defn \max( | u | - \lambda, 0 ) \frac{u}{ | u | } = u - \min(\lambda, |u|) \text{sgn}(u).
\end{equation*}
Since soft-thresholding only affects the difference components, it can be written as
\begin{align*}
    \approxTV & = \Wbm^{\Tsf} \Tcal^{\hspace{0.15em} a}_{\tau}(\Wbm \zbm)  \\
        & = \Wbm^{\Tsf}\left( \Wbm \zbm - \begin{bmatrix}
        0 \\
        \vdots \\
        0 \\
        \min \left(\tau 2 \sqrt{d}, | [\Wbm \zbm]^{\text{dif}}_1 | \right) \text{sgn}\left([\Wbm \zbm]^{\text{dif}}_1 \right) \\
        \vdots \\
        \min\left(\tau 2 \sqrt{d}, | [\Wbm \zbm]^{\text{dif}}_{nd} | \right) \text{sgn}\left([\Wbm \zbm]^{\text{dif}}_{nd}\right)
    \end{bmatrix}  \right) \\
    & = \zbm - \left[ \frac{\Mbm^{\Tsf}}{2 \sqrt{d}} \frac{\Dbm^{\Tsf}}{2 \sqrt{d}}
     \right] \begin{bmatrix}
        0 \\
        \vdots \\
        0 \\
        \min \left(\tau 2 \sqrt{d}, \frac{| [ \Dbm \zbm]_1 |}{2 \sqrt{d}} \right) \text{sgn}\left( \frac{[ \Dbm \zbm]_1}{2 \sqrt{d}} \right) \\
        \vdots \\
        \min\left(\tau 2 \sqrt{d}, \frac{| [\Dbm \zbm]_{nd} |}{2 \sqrt{d}}  \right) \text{sgn}\left( \frac{[ \Dbm \zbm]_{nd}}{2 \sqrt{d}} \right) 
    \end{bmatrix} \\
        & = \zbm - \frac{\Dbm^{\Tsf}}{2 \sqrt{d}} \begin{bmatrix}
        \min\left(\tau 2 \sqrt{d}, \frac{| [ \Dbm \zbm]_1 | }{2 \sqrt{d}} \right) \text{sgn}\left( \frac{[\Dbm \zbm]_1}{2 \sqrt{d}}   \right) \\
        \vdots \\
        \min \left(\tau 2 \sqrt{d}, \frac{| [\Dbm \zbm]_{nd} | }{2 \sqrt{d}} \right) \text{sgn}\left( \frac{[ \Dbm \zbm]_{nd}}{2 \sqrt{d}} \right)
    \end{bmatrix} \\
    & = \zbm - \frac{\Dbm^{\Tsf}}{2 \sqrt{d}} \min \left(\tau 2 \sqrt{d}, \frac{| \Dbm \zbm |}{2 \sqrt{d}}  \right) \text{sgn} \left( \frac{\Dbm \zbm}{2 \sqrt{d}}  \right) \\
    & = \zbm - \tau \Dbm^{\Tsf} \min \left( 1, \frac{| \Dbm \zbm |}{\tau 4 d}  \right) \text{sgn} \left( \Dbm \zbm \right) .
\end{align*}
In the last two lines, the operations $\min$ and $\text{sgn}$ are to be understood as being applied element-wise to $\Dbm \zbm$.
The function $\phi_{\theta}$ from 
Eq.~\eqref{eq: huberized anisotropic helper}
has the derivative,
\begin{equation*}
    \phi_{\theta}' (u) = \begin{cases}
        \frac{u}{\theta} & \text{for } |u| \leq \theta \\
        \text{sgn}(u) & \text{for } |u| > \theta 
    \end{cases} = \min \left(1, \frac{|u|}{\theta} \right) \text{sgn}( u )
\end{equation*}
Therefore, the operator can be written as 
$$\approxTV[a] = \zbm - \tau \Dbm^{\Tsf} \phi'_{\tau 4 d} \left( \Dbm \zbm \right),$$
where $\phi'_{\tau 4 d}$ is applied element-wise. 
In other words, $\approxTV[a]$ is a gradient step with step-size $\tau$ on 
\begin{equation} \label{eq: huberized anisotropic rewritten}
   \widetilde{h}^a(\zbm) = \sum_{i=1}^{nd} \phi_{\tau 4 d}( [\Dbm \zbm ]_i).
\end{equation}
A step-size of $\tau$ satisfies the criterion for the step size to be $\frac{1}{L}$, where $L$ is the Lipschitz constant of the gradients.
The Lipschitz constant for Eq.~\ref{eq: huberized anisotropic rewritten} is $1/\tau$, which can be verified as follows:
\begin{align*}
    \| \nabla \widetilde{h}(\xbm) - \nabla \widetilde{h}(\ybm) \| & = \| \Dbm^{\Tsf} \phi'_{4 \tau d}(\Dbm \xbm) - \Dbm^{\Tsf} \phi'_{4 \tau d}(\Dbm \ybm) \| \\
        & \leq \| \Dbm \| \| \phi'_{4 \tau d}(\Dbm \xbm) - \phi'_{4 \tau d}(\Dbm \ybm) \|  \\
        & \leq \left(\frac{1}{4 \tau d}\right) \| \Dbm \|  \| \Dbm \xbm - \Dbm \ybm \| \\
        & \leq \left(\frac{1}{4 \tau d}\right) \| \Dbm \|^2  \| \xbm -  \ybm \| \\
        & = \frac{1}{\tau}   \| \xbm -  \ybm \| ,
\end{align*}
where we use the fact that the finite difference matrix $\Dbm$ has spectral norm $2 \sqrt{d}$.

Now consider isotropic $\approxTV[i]$.
The soft-thresholding operator can be rewritten as
\begin{equation*} 
    \shrinkageFunction[\hspace{0.15em} \text{i}][\lambda]{\ubm}  := \max( \| \ubm \|_2 - \lambda, 0 ) \frac{\ubm}{ \| \ubm \|_2 } = \ubm - \min(\lambda, \|\ubm\|)\frac{\ubm}{\|\ubm\|}.
\end{equation*}
Let $G_i$ for $i \in \{1,...,n\}$ be the set of $d$ indices corresponding to the differences along all $d$ dimensions for $x_i$.
Since soft-thresholding affects only the difference components, it can be rewritten
\begin{align*}
    \approxTV &= \Wbm^{\Tsf} \Tcal^{\hspace{0.15em} i}_{\tau}(\Wbm \zbm)  \\
        & = \zbm - \frac{\Dbm^{\Tsf}}{2 \sqrt{d}} \begin{bmatrix}
        \min\left(\tau 2 \sqrt{d}, \frac{ \| [ \Dbm \zbm]_{G_1} \| }{2 \sqrt{d}} \right) \frac{[ \Dbm \zbm]_{G_{1}}}{  \|[ \Dbm \zbm]_{G_{1}} \|}\\
        \vdots \\
        \min \left(\tau 2 \sqrt{d}, \frac{ \| [\Dbm \zbm]_{G_{n}} \| }{2 \sqrt{d}} \right)  \frac{[ \Dbm \zbm]_{G_{n}}}{  \|[ \Dbm \zbm]_{G_{n}} \|}
    \end{bmatrix} \\
    & = \zbm - \tau \Dbm^{\Tsf} 
        \min \left(1, \frac{ \| [\Dbm \zbm]_{G_{i}} \| }{\tau 4d} \right)  \frac{[ \Dbm \zbm]_{G_{i}}}{  \|[ \Dbm \zbm]_{G_{i}} \|}.
\end{align*}
The last line is meant to be understood as applying $\min$ and  the norm to each grouping of indices.
The function $\psi_{\theta}$ from Eq.~\eqref{eq: huberized isotropic helper} has the gradient
\begin{equation*}
    \nabla \psi_{\theta} (\ubm) = \begin{cases}
        \frac{\ubm }{ \theta} & \text{for } \| \ubm \| \leq \theta \\
        \frac{\ubm}{\| \ubm \|} & \text{for } \| \ubm \| > \theta
    \end{cases} = \min\left( 1, \frac{\| \ubm \| }{\theta} \right) \frac{\ubm}{\| \ubm \|}
\end{equation*}
and so the operator can be written as 
\begin{equation*}
    \approxTV[i] = \zbm - \tau \Dbm^{\Tsf} \nabla \phi_{\tau 4 d}(\Dbm \zbm),
\end{equation*}
where $\nabla \phi_{\tau 4 d}$ is applied to each grouping of indices.
That is, $\approxTV[i]$ is a gradient step with the step-size $\tau$ on 
\begin{equation}  \label{eq: huberized isotropic rewritten}
   \widetilde{h}^i(\zbm) = \sum_{i=1}^{n} \phi_{\tau 4 d}( [\Dbm \zbm ]_{G_i}).
\end{equation}
Again, $\tau$ satisfies the step-size to be $\frac{1}{L}$, since the Lipschitz constant of the gradients for Eq.~\ref{eq: huberized isotropic rewritten} remains $1/\tau$.
\end{proof}

\section*{Proof of Proposition~\ref{prop: main prop}}
\begin{repproposition}{prop: main prop}
    For all $\zbm \in \R^{\sizeConst}$, the operator $\approxTV$ satisfies
    \begin{enumerate}[label=(\alph*)]
        \item $\approxTV = \prox_{\tau \TV}(\zbm + \deltabm)$ where $\| \deltabm \| \leq \tau \operatorBoundConst$, 
        \vspace{0.5em}
        \item $\zbm - \approxTV \in \tau \partial_{\tau \operatorSubdiffConst} \TV(\approxTV)$,
    \end{enumerate}
    where $\epsilon_1$ and $\epsilon_2$ are constants.
\end{repproposition}

\begin{proof}
    \textbf{Proof of Part (a).}
    We know $\|\Wbm \|_2 = 1$ from \eqref{eq: spectral norm W^T}.
    For both the anisotropic and isotropic versions, by its definition the soft-thresholding operator 
    $\shrinkageFunction[][\tau \hspace{0.15em} \shrinkStepConst]{\Wbm \zbm}$
    at most changes each of the $\sizeConst \dimConst$ difference elements of $\Wbm \zbm$ by $2\tau\sqrt{\dimConst}$, so we can bound 
    \begin{align} 
        \|\shrinkageFunction[][\tau \hspace{0.15em} \shrinkStepConst]{\Wbm \zbm} - \Wbm \zbm\|_2 & \leq \sqrt{\sizeConst \dimConst(2\tau\sqrt{\dimConst})^2} \nonumber \\
        & = \sqrt{4\tau^2 \sizeConst \dimConst^2} \nonumber \\
        & = 2\tau \dimConst \sqrt{\sizeConst} \;. \label{eq: soft thresh bound}
    \end{align}
    We bound $\| \approxTV - \zbm \|_2$ by  Eq.~\eqref{eq: spectral norm W^T} and Eq.~\eqref{eq: soft thresh bound} 
    \begin{align} 
         \Big\| \approxTV - \zbm \Big\|_2 & = \Big\| \approxTV - \Wbm^{\Tsf} \Wbm \zbm \Big\|_2 \nonumber \\
         & \leq \Big\|\shrinkageFunction[][\tau \hspace{0.15em} \shrinkStepConst]{\Wbm \zbm} - \Wbm \zbm \Big\|_2 \nonumber \\
         & \leq 2\tau \dimConst \sqrt{\sizeConst} \;. \label{eq: operator max distance}
    \end{align}
    For $\Phi_{\tau}(\ybm) :=  \frac{1}{2\tau}\|\ybm - \zbm \|_2^2 +\TV(\ybm)$,
    its subdifferential at $\approxTV$ is 
    \begin{align*}
        \partial \Phi_{\tau}(\approxTV) &  = \frac{1}{\tau}(\approxTV - \zbm) + \partial \TV(\approxTV) \\
        &  = \frac{1}{\tau}(\approxTV - \zbm) + \partial [\helperh \circ \Wbm](\approxTV) \\
        &  = \frac{1}{\tau}(\approxTV - \zbm) + \Wbm^{\Tsf} \partial   \helperh(\Wbm \approxTV) \;,
    \end{align*}
    where the second equality comes from $\helperh \circ \Wbm=\TV$ and the third equality comes from the chain rule.
    
    The distance $d(\bm{0}, \partial \Phi_{\tau}(\approxTV)$ is defined as the minimum $\ell_2$ distance between $\bm{0}$ and any vector in $\partial \Phi_{\tau}(\approxTV)$.
    Therefore, letting $\gbm_1 \in \partial  \helperh(\Wbm \approxTV)$, we get the following inequality
    \begin{align}
        d(\bm{0}, \partial \Phi_{\tau}(\approxTV) 
        &  \leq \Big\| \bm{0} -\left( \frac{1}{\tau} (\approxTV - \zbm) +  \Wbm^{\Tsf}\gbm_1 \right) \Big\|_2 \nonumber \\
        &  \leq \frac{1}{\tau}\Big\| \approxTV - \zbm \Big\|_2 + \Big\| \Wbm^{\Tsf}\gbm_1 \Big\|_2 \nonumber \\
        &  \leq \frac{1}{\tau}(2 \tau \dimConst \sqrt{\sizeConst}) + \| \gbm_1 \|_2 \nonumber \\
        &   \leq \frac{1}{\tau}(2 \tau \dimConst \sqrt{\sizeConst}) + 2 \dimConst \sqrt{\sizeConst} = \frac{ \tau \operatorBoundConst}{\tau} \;, \label{eq: type 3 ineq}
    \end{align}
    with $\operatorBoundConst=4 \dimConst \sqrt{\sizeConst}$.
    The third inequality comes from \eqref{eq: operator max distance} and nonexpansiveness of $\Wbm^{\Tsf}$.
    The fourth inequality comes from the fact that the norm of any gradient of $\helperh$ is bounded by $2\dimConst \sqrt{\sizeConst}$, as shown in Lemma~\ref{lemma: bounded subgradients}.
    
    Letting $\gbm_2 \in \partial \Phi_{\tau}(\approxTV)$, we can write $\| \gbm_2 \| \leq \frac{\tau \operatorBoundConst}{\tau}$.
    Equivalently, 
    \begin{gather*}
        \gbm_2 \in \frac{1}{\tau}(\approxTV - \zbm) + \partial \TV(\approxTV) \\
            \iff \\
        \zbm + \deltabm - \approxTV \in \tau \partial  \TV(\approxTV) \\
            \iff \\
            \approxTV = \prox_{\tau \TV}(\zbm + \deltabm), \text{ where } \| \deltabm \|_2 \leq \tau \operatorBoundConst \;, 
    \end{gather*}
    and  $\deltabm = \tau \gbm_2$.
    The last equivalence comes from the property of proximal operators that $\xbm = \prox_{\tau \TV}(\ybm)  \iff \ybm - \xbm \in \tau \partial \hspace{0.15em} \TV( \xbm)$.
    This establishes the desired result of part (b).
    
    \textbf{Proof of Part (b).}
    From the definition of the subgradient and the proximal operator, we have
    \begin{align*}
         \ubm = \prox_{\tau \helperh}(\Wbm \zbm) \quad \iff \quad \Wbm \zbm - \ubm \in \tau \partial \helperh(\ubm) \;. 
    \end{align*}
    By definition of subgradient, $\forall \wbm \in \R^{2 \sizeConst d}$
    \begin{align}
        \helperh(\wbm) & \geq \helperh(\ubm) + \left(\frac{\Wbm \zbm - \ubm}{\tau}\right)^{\Tsf}(\wbm - \ubm) \;. \nonumber
    \end{align}
    By replacing $\wbm$ with $\Wbm \ybm$, where $\ybm \in \R^{\sizeConst}$, we obtain
    \begin{align}
         \helperh(\Wbm \ybm) & \geq \helperh(\ubm) + \frac{1}{\tau} (\Wbm \zbm - \ubm)^{\Tsf}(\Wbm \ybm - \ubm) \nonumber   \\
         & = \helperh(\ubm) + \frac{1}{\tau} \Big( (\Wbm\zbm)^{\Tsf}(\Wbm \ybm)  - (\Wbm(\zbm + \ybm))^{\Tsf} \ubm + \|\ubm\|_2^2 \Big) \nonumber \\
         & \geq \helperh(\ubm) + \frac{1}{\tau} \Big( \zbm^{\Tsf}\ybm  - (\zbm + \ybm)^{\Tsf} (\Wbm^{\Tsf}\ubm) + \| \Wbm^{\Tsf} \ubm\|_2^2 \Big) \nonumber \\
         & = \helperh(\ubm) + \frac{1}{\tau} (\zbm - \Wbm^{\Tsf}\ubm)^{\Tsf} (\ybm - \Wbm^{\Tsf}\ubm )  \nonumber   \\
         & = \helperh(\shrinkageFunction[][\tau \hspace{0.15em} \shrinkStepConst]{\Wbm \zbm})  + \frac{1}{\tau} (\zbm - \approxTV)^{\Tsf} (\ybm - \approxTV ) \;, \label{eq: type 2 bound from def of subgrad}
    \end{align}
    where the last inequality comes from the definition of adjoint $(\Wbm \xbm_1)^{\Tsf}\xbm_2 = \xbm_1^{\Tsf} ( \Wbm^{\Tsf}\xbm_2 )$, the fact that $\Wbm^{\Tsf}\Wbm = \Ibm$, and that $\|\Wbm^T\|_2 = 1$. 
    In the last line, we used the fact that $ \ubm = \prox_{\tau \helperh}(\Wbm \zbm) = \shrinkageFunction[][\tau \hspace{0.15em} \shrinkStepConst]{\Wbm \zbm}$ from Eq.~\ref{eq: prox of hhat} and $\approxTV = \Wbm^{\Tsf} \ubm$  . 

\begin{figure}[t]
 \centering
 \includegraphics[scale=1]{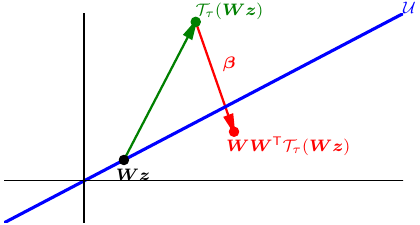}
 \caption{
 A visual illustration of the characterization of the projection as $\Wbm\Wbm^{\Tsf} \shrinkageFunction[][\tau \hspace{0.15em} \shrinkStepConst]{\Wbm \zbm} = \shrinkageFunction[][\tau \hspace{0.15em} \shrinkStepConst]{\Wbm \zbm} + \deltabm$.
 For any $\zbm \in \R^n$, the vector $\Wbm \zbm \in \Ucal := \{\ubm | \ \ubm = \Wbm \Wbm^{\Tsf} \ubm \}$.
 Applying the soft-thresholding function gives $\Tcal_{\tau \hspace{0.15em} \shrinkStepConst}(\Wbm \zbm)$, not necessarily in $\Ucal$. 
 Applying $\Wbm \Wbm^{\Tsf}$ projects back to $\Ucal$, which can be represented as $\shrinkageFunction[][\tau \hspace{0.15em} \shrinkStepConst]{\Wbm \zbm} + \deltabm$.
 }
\label{fig: projection visualization}
\end{figure}
    The projection can be characterized as
    \begin{equation}
        \Wbm\Wbm^{\Tsf} \shrinkageFunction[][\tau \hspace{0.15em} \shrinkStepConst]{\Wbm \zbm} = \shrinkageFunction[][\tau \hspace{0.15em} \shrinkStepConst]{\Wbm \zbm} + \betabm \;.
    \end{equation}
    That is, $\betabm$ is the orthogonal component of $\shrinkageFunction[][\tau \hspace{0.15em} \shrinkStepConst]{\Wbm \zbm}$ to the subspace $\Ucal := \{\ubm | \ \ubm = \Wbm \Wbm^{\Tsf} \ubm \}$.
    Figure~\ref{fig: projection visualization} provides a visual representation.
    
    Since the projection $\Wbm \Wbm^{\Tsf}$ minimizes the $\ell_2$ distance between $\shrinkageFunction[][\tau \hspace{0.15em} \shrinkStepConst]{\Wbm \zbm}$ and $\Ucal$, the following inequality holds
    \begin{align*}
        \| \betabm \|_2  
        & = \min_{\wbm \in \Ucal} \| \wbm - \shrinkageFunction[][\tau \hspace{0.15em} \shrinkStepConst]{\Wbm \zbm} \|_2 \\
        & \leq  \| \wbm - \shrinkageFunction[][\tau \hspace{0.15em} \shrinkStepConst]{\Wbm \zbm} \|_2 \qquad \forall \wbm \in \Ucal \;.
    \end{align*}
    Since $\Wbm \zbm \in \Ucal$, this yields
    \begin{equation*}
        \|\betabm\|_2 \leq \| \shrinkageFunction[][\tau \hspace{0.15em} \shrinkStepConst]{\Wbm \zbm} - \Wbm \zbm\|_2 \leq 2\tau \dimConst \sqrt{\sizeConst} \;,
    \end{equation*}
    where the last inequality comes from \eqref{eq: soft thresh bound}.
    Letting $\betabm^{\text{dif}}$ be the components of $\deltabm$ corresponding to difference coefficients,
    \begin{align*}
        \| \betabm^{\text{dif}} \|_{2,1} \leq \| \betabm^{\text{dif}} \|_{1,1} & = \| \betabm^{\text{dif}} \|_1 \\ 
            & \leq \sqrt{\sizeConst \dimConst} \| \betabm^{\text{dif}} \|_2 \\
            & \leq \sqrt{\sizeConst \dimConst} \| \betabm \|_2 \\ 
            & \leq (\sqrt{\sizeConst \dimConst})(2\tau \dimConst \sqrt{\sizeConst}) \\
            & = 2 \tau \sizeConst \dimConst^{3/2} \;,
    \end{align*}
    where the first inequality comes from the fact that the $\ell_2$ norm is smaller than the $\ell_1$ norm for any given vector.
    The second inequality comes from the fact that the $\ell_1$ norm is at most the square root of the dimension of the vector times its $\ell_2$ norm.
    
    Using the definition of $\helperh$ in Eq.~\eqref{eq: helper h definition} relationship $\helperh(\Wbm \zbm) = \TV(\zbm)$, shown in Eq.~\eqref{eq: helperh related to tv}, we get
    \begin{align}
        \TV(\approxTV) & = \helperh(\Wbm \approxTV) \nonumber \\
        & = 2\sqrt{\dimConst} \| [\Wbm \approxTV ]^{\text{dif}} \|_{p,1} \nonumber \\
        & = 2\sqrt{\dimConst} \| [\Wbm \Wbm^{\Tsf} \shrinkageFunction[][\tau \hspace{0.15em} \shrinkStepConst]{\Wbm \zbm}]^{\text{dif}} \|_{p,1} \nonumber \\ 
        & = 2\sqrt{\dimConst} \| [\shrinkageFunction[][\tau \hspace{0.15em} \shrinkStepConst]{\Wbm \zbm} + \betabm]^{\text{dif}} \|_{p,1} \nonumber \\ 
        & \leq 2\sqrt{\dimConst}\| [\shrinkageFunction[][\tau \hspace{0.15em} \shrinkStepConst]{\Wbm \zbm}]^{\text{dif}} \|_{p,1} \nonumber \\
            & \qquad \qquad \qquad + 2\sqrt{\dimConst}\| \betabm^{\text{dif}} \|_{p,1} \nonumber \\
        & \leq 2\sqrt{\dimConst}\| [\shrinkageFunction[][\tau \hspace{0.15em} \shrinkStepConst]{\Wbm \zbm}]^{\text{dif}} \|_{p,1} \nonumber \\
            & \qquad \qquad \qquad + (2\sqrt{\dimConst})(2 \tau \sizeConst \dimConst^{3/2}) \nonumber \\
        & = \helperh(\shrinkageFunction[][\tau \hspace{0.15em} \shrinkStepConst]{\Wbm \zbm}) + 4\tau \sizeConst \dimConst^2 \;. \nonumber %
    \end{align}
    Rearranging the above inequality yields
    \begin{align}
        \helperh(\shrinkageFunction[][\tau \hspace{0.15em} \shrinkStepConst]{\Wbm \zbm}) \geq \TV(\approxTV) - 4\tau \sizeConst \dimConst^2 \;. \label{eq: type 2 bound from triangle}
    \end{align}
    From Eq.~\eqref{eq: type 2 bound from def of subgrad} we have
    \begin{align*}
        \TV(\ybm) \geq \helperh(\shrinkageFunction[][\tau \hspace{0.15em} \shrinkStepConst]{\Wbm \zbm}) + \frac{1}{\tau} (\zbm - \approxTV)^{\Tsf} (\ybm - \approxTV ) \;.
    \end{align*}
    Combining this with Eq.~\eqref{eq: type 2 bound from triangle} provides the following for all $\ybm \in \R^{\sizeConst}$
    \begin{align*}
        \TV(\ybm)   \geq  \TV(\approxTV)  + \frac{1}{\tau} (\zbm-\approxTV)^{\Tsf} (\ybm - \approxTV) - \tau 4 \sizeConst \dimConst^2 \;.
    \end{align*}
    By Definition~\ref{def: epsilon subdifferential} this establishes
    \begin{align*}
        \zbm - \approxTV \in \tau \partial_{\tau \operatorSubdiffConst } \TV( \approxTV )\;,
    \end{align*}
    where $\operatorSubdiffConst = 4 \sizeConst \dimConst^2$.
    This establishes the desired result for Part (c).
\end{proof}

\begin{lemma} \label{lemma: bounded subgradients}
    For both the anisotropic and isotropic versions, any subgradient $\gbm \in \partial \helperh(\ubm)$ has a bounded $\ell_2$ norm $\| \gbm \|_2 \leq 2 \dimConst \sqrt{\sizeConst}$.
\end{lemma}
\begin{proof}
    This lemma makes use of the following subdifferential sum property: 
    for any functions $f_1, \dots, f_m \in \Gamma^0$ such that the sets $\text{ri}(\text{dom} f_i)$ have a point in common, where $\text{ri}$ is the relative interior, then 
    \begin{align} \label{eq: subdifferential sum}
        \partial \Bigg[\sum_{i=1}^m f_i \Bigg](\xbm) = \sum_{i=1}^m \partial f_i(\xbm) \qquad \forall \xbm \in \R^{\sizeConst} \;.
    \end{align}
    See Theorem 23.8 in~\cite{Rockafellar.1970} and Corollary 16.48 in~\cite{BauschkeCombettes2017}.
    First consider the anisotropic version
    \begin{equation*}
        \helperh[\text{a}](\ubm) = 2\sqrt{\dimConst}\| \ubm^{\text{dif}} \|_{1,1} = 2\sqrt{\dimConst} \sum_{i=1}^{\sizeConst} \| [\ubm^{\text{dif}}]_{i} \|_1 \;.
    \end{equation*}
    Since the domain of each $|| \cdot||_1$ is all of $\R^{2nd}$, the relative interiors of the domains have a nonempty intersection, and so we can apply Eq.~\eqref{eq: subdifferential sum}.
    Consider $\partial \big( \| [\ubm^{\text{dif}}]_{i} \|_1 \big)$, whose norm is bounded by $\sqrt{d}$.
    By Eq.~\eqref{eq: subdifferential sum}, 
    letting $\gbm \in \partial \helperh[\text{a}](\ubm)$, we can bound
    \begin{align*}
        \| \gbm \|_2 & = 2\sqrt{\dimConst} \sqrt{ \sum_{i=1}^{\sizeConst}  \Big( \partial \big( \| [\ubm^{\text{dif}}]_{i} \|_1 \big)\Big)^2 } \\
        & \leq (2\sqrt{\dimConst}) \sqrt{\sizeConst (\sqrt{\dimConst})^2} \\
        &= 2\dimConst \sqrt{\sizeConst} \;.
    \end{align*}
    Now consider the isotropic version
    \begin{equation*}
         \helperh[\text{i}](\ubm) = 2\sqrt{\dimConst}\| \ubm^{\text{dif}} \|_{2,1} = 2\sqrt{\dimConst} \sum_{i=1}^{\sizeConst} \| [\ubm^{\text{dif}}]_{i} \|_2 \;.
    \end{equation*}
    Similarly, the domain of each $|| \cdot||_2$ is all of $\R^{2nd}$, and so we can apply Eq.~\eqref{eq: subdifferential sum}.
    Consider $\partial \big( \|[\ubm^{\text{dif}}]_n \|_2 \big)$, whose norm is bounded by $1$.
    By Eq.~\eqref{eq: subdifferential sum}, letting $\gbm \in \partial \helperh[\text{i}](\ubm)$, we obtain
    \begin{align*}
        \| \gbm \|_2 = 2 \sqrt{\dimConst}  \sqrt{ \sum_{i=1}^n \Big\| \partial \Big( \|[\ubm^{\text{dif}}]_{i} \|_2 \Big) \Big\|_2^2 } \leq 2\sqrt{ \sizeConst \dimConst } \;.
    \end{align*}
    Since $\dimConst \geq 1$, we have $2\sqrt{ \sizeConst \dimConst } \leq 2 \dimConst \sqrt{ \sizeConst }$.
    Therefore, for both the anisotropic and isotropic, $\|\gbm\|_2$ is bounded by $2\dimConst \sqrt{\sizeConst}$. 
\end{proof}

\bibliographystyle{IEEEbib}
\bibliography{references}

\end{document}